\documentclass[conference]{IEEEtran}

\usepackage[utf8]{inputenc}
\usepackage{cite}
\usepackage{amsmath,amssymb,amsfonts}
\usepackage{graphicx}
\usepackage{mathtools}
\usepackage{inconsolata}
\usepackage{textcomp}
\usepackage{xcolor}
\usepackage{xspace}
\usepackage{booktabs}
\usepackage{comment}
\usepackage[hidelinks]{hyperref}
\def\BibTeX{{\rm B\kern-.05em{\sc i\kern-.025em b}\kern-.08em
    T\kern-.1667em\lower.7ex\hbox{E}\kern-.125emX}}

\usepackage{setspace}
\usepackage{listings}

\usepackage{amsthm}

\newtheorem{property}{Property}[section]

\newtheorem{lemma}{Lemma}[section]

\newtheorem{theorem}{Theorem}

\newcommand{\subparagraph}{}

\usepackage[ruled]{algorithm2e}
\SetKwProg{Fn}{function}{}{end}
\SetKwProg{Init}{initialization}{}{}
\SetKwInput{Input}{input}
\SetKw{KwTo}{in}
\SetKw{KwWhere}{where}
\newcommand{\squeezeup}{\vspace{-5mm}}

\DeclareMathOperator*{\argmax}{arg\,max}

\newcommand{\ie}{i.e.,\xspace}
\newcommand{\eg}{e.g.,\xspace}

\newcommand{\eat}[1]{}

\definecolor{darkgreen}{rgb}{0.0, 0.5, 0.0}

\newcommand\code[1]{{\small\lstinline$#1$}}

\newcommand\Reals{\mathbb{R}}
\newcommand{\norm}[1]{\left\lVert#1\right\rVert}

\def\simdex{\textsc{Maximus}\xspace}
\def\opt{\textsc{Optimus}\xspace}
\def\simdexi{\simdex}
\newcommand\sbi{\simdex}
\def\simdexb{blocked matrix multiply\xspace}
\def\U{\ensuremath{U}\xspace}
\def\I{\ensuremath{I}\xspace}
\def\u{\ensuremath{\mathbf{u}}\xspace}
\def\j{\mathbf{j}\xspace}
\def\c{\mathbf{c}}
\def\topK{top-$K$\xspace}
\def\topk{\topK}
\newcommand{\papername}{To Index or Not to Index:\\Optimizing Exact Maximum Inner Product Search}

\newcommand{\minihead}[1]{{\vspace{.45em}\noindent\textbf{#1.} }}
\newcommand{\miniheadd}[1]{{\vspace{.45em}\noindent\textbf{#1} }}
\newcommand{\miniheadit}[1]{{\vspace{.45em}\noindent\textit{#1.} }}

\begin{document}

\author{\IEEEauthorblockN{Firas Abuzaid, Geet Sethi, Peter Bailis, Matei Zaharia}
\IEEEauthorblockA{\textit{Stanford DAWN Project}}
}

\title{\papername}
\maketitle

\begin{abstract}
Exact Maximum Inner Product Search (MIPS) is an important task that is widely
pertinent to recommender systems and high-dimensional similarity search. The
brute-force approach to solving exact MIPS is computationally expensive, thus
spurring recent development of novel indexes and pruning techniques for this
task. In this paper, we show that a hardware-efficient brute-force approach,
blocked matrix multiply (BMM), can outperform the state-of-the-art MIPS
solvers by over an order of magnitude, for some---but not all---inputs.

In this paper we also present a novel MIPS solution, \simdex, that takes
  advantage of hardware efficiency {\em and} pruning of the search space. Like
  BMM, \simdex is faster than other solvers by up to an order of magnitude, but
  again only for some inputs. Since no single solution offers the best runtime
  performance for {\em all} inputs, we introduce a new data-dependent
  optimizer, \opt, that selects online with minimal overhead the best MIPS
  solver for a given input. Together, \opt and \simdex outperform
  state-of-the-art MIPS solvers by 3.2$\times$ on average, and up to
  10.9$\times$, on widely studied MIPS datasets.
\end{abstract}

\section{Introduction}
\label{sec:intro}



Over the last decade, the Maximum Inner Product Search (MIPS)
problem has increasingly become more important in machine learning and large-scale
data systems. The classical MIPS setup is as follows: given a vector $\u \in
\Reals^f$ and a second set of vectors $I$, where each $\mathbf{i} \in I$ is
also in $\Reals^f$, find the vector $\mathbf{i}$ that maximizes the inner
product $\u^T\mathbf{i}$.  The problem generalizes even further, to finding the
top $K$ vectors in $I$ that maximize $\u$.

MIPS has attracted interest because it is relevant to many real-world
applications. For example, recommender systems---which are deployed in
e-commerce~\cite{linden2003amazon}, social networking~\cite{facebookblog}, and
media applications~\cite{koren2009matrix}---are often based on latent factor
modeling using matrix factorization
(MF)~\cite{wang2015collaborative,van2013deep}. In an MF model, each user is
associated with a vector $\u$, and each item (\eg a movie or song) is
associated with a vector $\mathbf{i}$. The predicted rating for the item by the user is
modeled as $\u^T\mathbf{i}$; therefore, solving MIPS in the context of an MF model
effectively computes the \topK recommendations for a given user. Similarly,
MIPS can be applied to high-dimensional similarity search~\cite{glove} and
multi-class prediction tasks~\cite{imagenet}, too.


For exact MIPS in the batch setting (\ie compute \topK for all users at once),
the na\"ive approach fails to scale to today's data volumes, and two indexing
techniques recently developed by the database
community---LEMP~\cite{teflioudi2015lemp} and
FEXIPRO~\cite{fexipro}---have proven to be the most efficient MIPS solvers and
currently lead in terms of performance. The efficiency of these indexes is
predicated on their ability to \textit{i)} prune the search space of candidate
items, thus limiting the number of inner products computed, and \textit{ii)}
improve the efficiency of computing the inner products themselves.

However, in this paper, we show that both of these indexes do not always
outperform a \emph{hardware-efficient} brute-force approach: a dense
matrix multiply of the user and item matrices using optimized, cache-efficient Basic Linear
Algebra Subprogram (BLAS) libraries~\cite{blas,intel-mkl,openblas}. This approach, which we refer
to as blocked matrix multiply (BMM) throughout this paper, has been applied to
other areas of machine learning (such as model
search~\cite{sparks2015automating}) and is surprisingly competitive with the
state of the art.
Prior indexing techniques for MIPS have not taken into account hardware
optimizations that could be adopted by brute-force methods; in fact, they sometimes
actually lose to the brute-force ones.

To illustrate this point, we trained 20 MF models
on three gold-standard benchmark datasets and found that brute-force
computation can outperform indexing on many of them, including several of the
most accurate models.  Figure~\ref{fig:appetizer_no_simdex} shows two examples.
For the recommendation model trained on the Netflix Prize dataset (left), BMM
outperforms the LEMP and FEXIPRO
indexes by a factor of 1.9--3.1$\times$. However, for the Yahoo R2 MF model
(right), LEMP and FEXIPRO are 2-3.5$\times$ faster than BMM.


\begin{figure}[t]
  \includegraphics[width=\columnwidth]{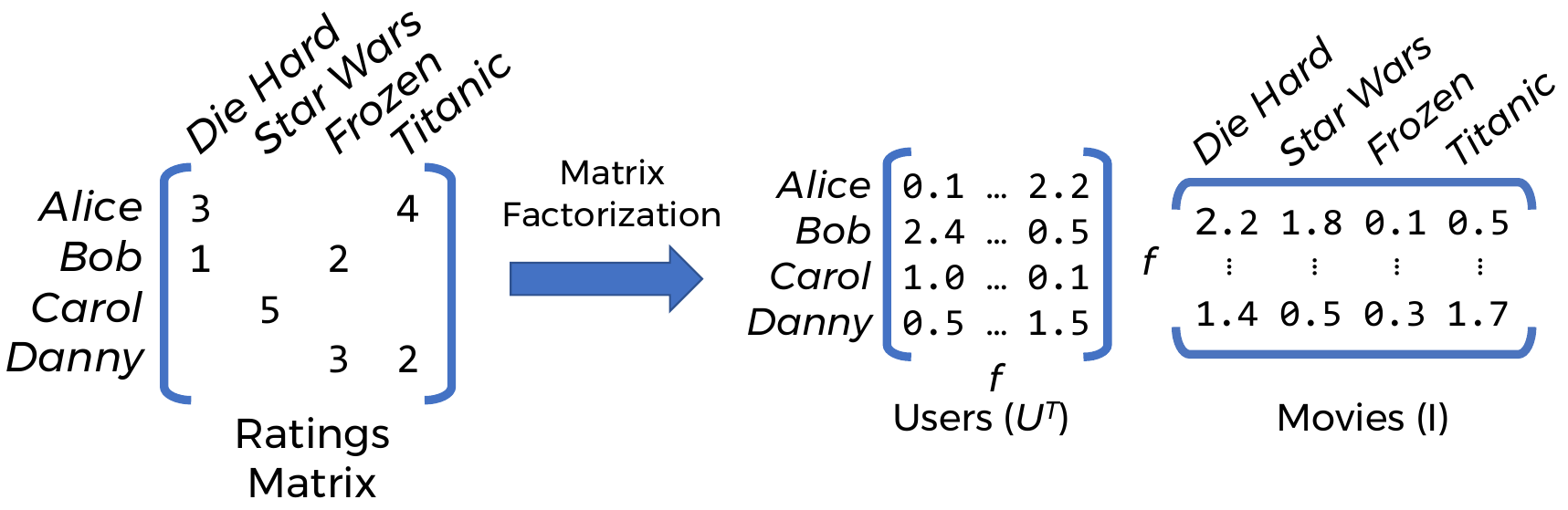}
  \squeezeup
  \caption{Example of a matrix factorization (MF) model, a common setting
  where MIPS is applied to find the \topK items in $\I$ (\ie movies) for each
  user $\u \in U$.  In this paper, we focus on accelerating the end-to-end MIPS
  runtime for all users in $U$.}
  \label{fig:matrix_factorization}
  \squeezeup
\end{figure}

Motivated by these experiments and others in Section~\ref{sec:evaluation}, we
propose a new, hardware-optimized index called \simdex, which captures the
hardware efficiency of BMM while also pruning the search space, \`a la LEMP and
FEXIPRO.  \simdex utilizes a combination of clustering, blocking, and linear
algebra primitives from hardware-efficient software packages~\cite{armadillo,
intel-mkl} to achieve high performance while still being easy to implement.
\simdex first clusters users based on their weight vectors using $k$-means,
then uses the cluster centroids to produce a sorted list of preferred
items for each cluster. The cluster's sorted list approximates the user's
preferred items for every user assigned to the cluster. Then, rather than find
each user's top $K$ items one at a time, \simdex blocks the traversal of the
sorted item lists for multiple users at a time, and applies a bound derived by
Koenigstein et al.~\cite{koenigstein2012efficient} to correct for the
approximation. By combining hardware-efficient blocking with this pruning
technique, \simdex outperforms LEMP and FEXIPRO by 1.78$\times$ on average, and
up to 10.6$\times$.

Despite the improvements that \simdex gives us, it still does not beat BMM on
all inputs. For exact MIPS, there is no single clear winner; a data-dependent
strategy is needed. Therefore, we also propose \opt, a MIPS serving optimizer
that automatically selects an efficient solving strategy for a particular
input.  To enable \opt to efficiently select the best strategy, we exploit two
key observations. First, for the current MIPS indexes, index traversal is, in
fact, much more expensive than index construction. For example, we show that the runtime to compute
even the top $K=1$ recommendation for all users can be 90$\times$
greater than the runtime to build an index. Thus, we can cheaply construct an index to test
its efficacy. Second, the performance of BMM and these MIPS indexes
can be determined by measuring the runtime on a small number of users, then
accurately extrapolating the performance on the entire model. This inspires a
sampling-based approach, in which we construct a MIPS index and then directly
compare its performance to BMM's on a small sample of users. By
applying an incremental t-test during this comparison, we can, for certain
inputs, determine the best serving strategy without exploring the full sample,
thereby reducing the runtime overhead. Overall, this online tuning step incurs
low overhead (5.5\% on average) while delivering high accuracy and large
speedups across four different index types (up to 6$\times$) that perform
within 12\% of an oracle-based optimizer with no overhead. These speedups are
especially surprising given the highly-optimized nature of each index---for
example, in their recent paper, LEMP improved performance over the previous
best alternative exact index by up to 24$\times$~\cite{teflioudi2015lemp}.

\begin{figure}[t]
  \includegraphics[width=\linewidth]{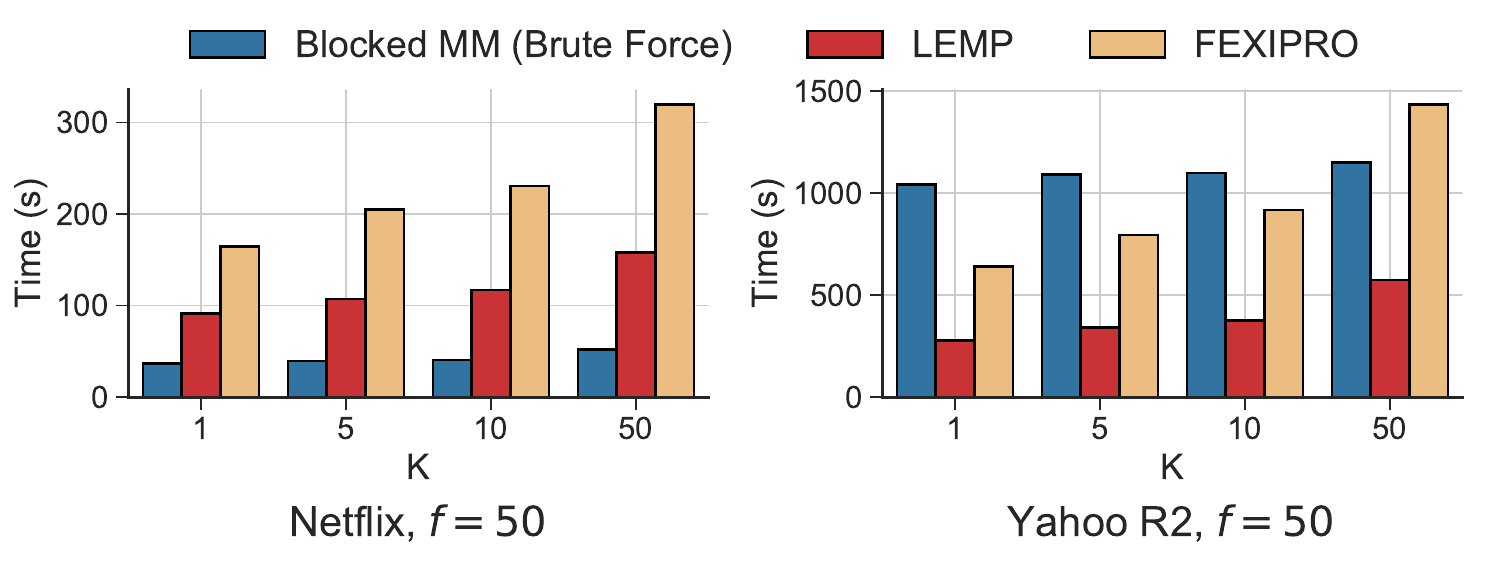}
  \squeezeup
  \caption{Experiment comparing LEMP~\cite{teflioudi2015lemp} and
  FEXIPRO~\cite{fexipro}, two state-of-the-art MIPS indexes, against a
  hardware-efficient brute-force approach, blocked matrix multiply. We
  measure the end-to-end \topK runtime for all users on two MF models from the
  Netflix Prize~\cite{netflix} and the Yahoo R2~\cite{r2} datasets.  For
  Netflix, blocked MM is surprisingly fastest; for Yahoo R2, LEMP and FEXIPRO outperform it. In this paper, we
  develop \opt, an optimizer that chooses online the best MIPS serving strategy
  for a target model.}
  \squeezeup
  \label{fig:appetizer_no_simdex}
\end{figure}

We show that combining \opt and \simdex yields a 3.2$\times$ speedup
on average (and up to 10.9$\times$) compared to the state-of-the-art
indexes alone. \opt and \simdex thus provide a pragmatic and
easy-to-implement but difficult-to-beat baseline for MIPS across
a wide range of inputs and query settings.

\section{Background}
\label{sec:background}

In this section, we provide background regarding the \topK MIPS problem, the
surprising competitiveness of blocked matrix multiply (BMM), and the current
state-of-the-art approaches to solving MIPS.

\subsection{Problem Statement: Batch Exact MIPS}
Let $U$ be a set of user vectors $\u \in \Reals^f$, and let $I$ be a set of
item vectors $\mathbf{i} \in \Reals^f$, where $f$ is referred to as the number of
\emph{latent factors}.  For a single vector \u, the Maximum Inner Product
Search problem is defined as

$$
  r_{ui} = \argmax_{\mathbf{i} \in \I}\ \u^T\mathbf{i}
$$

In this paper, we focus on the \emph{batch} \topK MIPS
problem~\cite{teflioudi2015lemp,shrivastava2014asymmetric,ram2012maximum},
where we wish to minimize the end-to-end runtime of finding the top
$K$ items that maximize the inner product for all $\u \in U$.
However, \simdexi, our proposed index, can also accelerate MIPS for
a subset of users at a time, as might happen in a model serving system like
Clipper~\cite{crankshaw2017clipper} that collects tens of requests at once.





\miniheadd{Why Exact MIPS?} 
While a range of existing techniques (Section~\ref{sec:related_work}) provide
solutions to the \emph{approximate} \topK retrieval problem, in this work, we
are interested in evaluating MIPS quickly, subject to serving the \emph{most
accurate} results available. For some domains,
especially revenue-critical applications that demand high-accuracy
recommendations (\eg e-commerce and advertising~\cite{schafer1999recommender}),
approximate techniques cannot be applied and exact MIPS is instead required.
(Notably, the margin of accuracy separating the first- and second-place teams in
the 2009 Netflix Prize challenge was only 0.01\%~\cite{bell2007lessons}.)

\subsection{Blocked Matrix Multiply}
The brute-force approach to MIPS is straightforward: for each user vector and
each item vector, compute the inner product between the two. Once the ratings
for all user-item pairs have been computed, select the top $K$ items for each
user (\eg using a min-heap).

Each inner product can be computed using \code{sdot}, a standard BLAS library
function~\cite{intel-mkl,openblas}. However, instead of repeatedly
calling \code{sdot} in a double for-loop over the user and item vectors, we can
replace the entire computation with a single matrix-matrix multiply between the
users matrix and the items matrix (\eg using the \code{sgemm} function in BLAS),
thereby blocking the entire computation of user-item ratings. We refer to this
approach as blocked matrix multiply (BMM) throughout this paper.



In theory, replacing inner products with a single matrix-matrix
multiplication does not change the runtime complexity of computing all
user-item ratings; both are brute-force. Thus, if an index filters out a
significant fraction of items, it should be faster. Therefore, why is
blocked matrix multiply often faster than indexes in practice?

The reason for this result is that, due to their popularity in numerical
workloads, matrix-matrix multiply kernels are highly optimized for modern hardware.
Specifically, modern linear algebra kernels will perform advanced data layout
and blocking to maximize cache utilization, and, when available (\eg on
modern server-class processors), will aggressively vectorize code using SIMD
instructions. These optimizations are applied automatically and transparently
to the end programmer in BLAS libraries. By not designing for
hardware efficiency, state-of-the-art MIPS indexes do not fully take advantage
of these benefits found in modern hardware.

Thus, when we perform blocked matrix multiply using one of many implementations
of linear algebra kernels, we benefit from decades of algorithmic and
hardware-specific optimizations that yield substantial empirical speedups over
na\"ive inner products ($40\times$) or even matrix-vector
multiply ($20\times$). 
These speedups act as an effective ``constant factor'' improvement in
runtime---they do not change the asymptotic performance, but, especially when
evaluating MIPS on reference inputs, these constants have a significant effect
on hardware efficiency.

\eat{
An alternative approach to computing the top $K$ that provides an
intermediate form of blocking would be a matrix-vector multiply between a
single user vector and the items matrix (\ie \code{sgemv} in BLAS). In effect,
this batches the item vectors, but not the users. For smaller models, this
approach is effective enough to meet the desired latency for some real-world
applications.  However, batching users can still provide efficiency wins even
for these small models: on Netflix, with 500K users, 20K items, and $f=100$, we
measured a 20$\times$ improvement in the per-user query time between blocked
matrix multiply and matrix-vector multiply.  This efficiency gap increases up
to 40$\times$ for our largest model GloVe-Twitter, with 100K users, 1M items,
and $f=200$.  Thus batching users is important to consider for computing MIPS
at scale, especially as more real-world applications move to even larger
models~\cite{,facebookblog}, and even online serving systems employ
batching to optimize inference~\cite{crankshaw2017clipper}.
}

\subsection{Existing MIPS Indexes}
\label{sec:background-indexes}
MIPS is a topic of active research
within the database community. In this section, we give an overview of prior
approaches used to solve MIPS, which also provides necessary context for the
remainder of this paper.

\minihead{User Clustering}
One of the first novel solutions to the MIPS problem was suggested by
Koenigstein et al.~\cite{koenigstein2012efficient}, who proposed
to cluster users via spherical clustering to perform \emph{approximate}
\topK queries. Effectively, they used the cluster centroids
to serve as an approximation of the users' preferences. To measure the approximation
error, they derived a bound on the distortion of ratings based on the angle
between the user vector and its assigned centroid. Here, we re-derive
the bound (Equation 13 in~\cite{koenigstein2012efficient}) as
Equation~\ref{eq:upper_bound} below and illustrate how to use this bound to
approximate \topK scoring.  (Note: This bound assumes that the user and item
vectors reside in a metric space, a safe assumption for exact MIPS.)

Suppose we have a user vector $\u$, an item vector $\mathbf{i}$, and a centroid vector
$\c$ that we have obtained from clustering the user vectors (\ie $\u$ is
assigned to the cluster represented by $\c$).  Let $\theta_{ui}$ be the angle
between $\u$ and $\mathbf{i}$, $\theta_{ic}$ be the angle between $\mathbf{i}$ and $\c$, and
$\theta_{uc}$ be the angle between $\u$ and $\c$. Finally, let $r_{ui}$ be the
rating for the user-item pair $\u,\mathbf{i}$.

By the triangle inequality on angular distances in metric spaces, we have that

$$
|\theta_{ic} - \theta_{uc}| \leq \theta_{ui} \leq \theta_{ic} + \theta_{uc}.
$$

\noindent Thus, we can compute an upper bound on $r_{ui}$:

$$
r_{ui} = \u^T \mathbf{i} = \norm{\u} \norm{\mathbf{i}} \cos(\theta_{ui}) \leq \norm{\u} \norm{\mathbf{i}} \max_{\theta \in [\theta_{ic} \pm \theta_{uc}]}\cos \ \theta.
$$

For \topK, each user's ratings are invariant under linear
scaling. Thus, we can omit $\norm{u}$ to preserve the \emph{relative}
ordering of items for a single user and obtain a linear scaling of
$r_{ui}$, denoted $r_{ui}^*$ such that $r_{ui}^*\norm{\u} = r_{ui}$:

\begin{equation}
\label{eq:max}
  r_{ui}^* \leq \norm{\mathbf{i}} \max_{\theta \in [\theta_{ic} \pm \theta_{uc}]} \cos \ \theta.
\end{equation}

Since $\cos^{-1}(x) \in [0, \pi]$, we have two cases to consider. If
$\theta_{uc} < \theta_{ic}$, then $\theta_{uc} < \pi$ and so
Equation~\ref{eq:max} is maximized at $\theta_{ic} - \theta_{uc}$. If
$\theta_{uc} \geq \theta_{ic}$, then Equation~\ref{eq:max} is
maximized at the origin. Rewriting Equation~\ref{eq:max}, we have:

\begin{equation}
  r_{ui}^* \leq
  \begin{cases}
    \norm{\mathbf{i}} \cos(\theta_{ic} - \theta_{uc}) & \text{if } \theta_{uc} < \theta_{ic} \\
      \norm{\mathbf{i}} & \text{otherwise}. \\
  \end{cases}
  \label{eq:upper_bound}
\end{equation}

Equation~\ref{eq:upper_bound} provides a means of monotonically ranking items
according to their angular distance from the cluster center. Later on, in
Section~\ref{sec:index}, we extend this inequality to perform \emph{exact}
queries using our proposed index, \simdexi.


\minihead{LEMP}
In SIGMOD 2015~\cite{teflioudi2015lemp} and TODS
2016~\cite{teflioudi2016lempTODS}, Teflioudi et al. introduced the LEMP index,
which empirically outperformed all prior approaches. LEMP solves the
MIPS problem using a divide-and-conquer approach: first, it sorts the item
vectors by length and partitions them into buckets, such that each bucket
contains vectors of roughly equal magnitude.
For each bucket, LEMP computes a set of candidate items for the top $K$ by solving a
smaller cosine similarity search problem; depending on the distribution of
vector lengths, LEMP can choose one of several possible algorithms to retrieve
these candidates, such as
length-based pruning, angular-based pruning, incremental pruning (which
leverages partial inner products and the Cauchy-Schwarz inequality), or na\"ive
search (\ie inner products). Crucially, LEMP does not consider blocked matrix
multiply when scoring multiple users. LEMP chooses the retrieval algorithm by
testing each method on a sample of user vectors.  Once each bucket has been
examined, LEMP merges the results by computing the ratings for each of the
candidate item vectors using inner products and then finally selects the top $K$
items. In an extensive study appearing in TODS 2016, LEMP was shown to
outperform exact indexing alternatives by up to 24$\times$.

\eat{
\minihead{FEXIPRO} In SIGMOD 2017~\cite{fexipro}, Li et al. introduced the
FEXIPRO index. FEXIPRO combines several pruning techniques to reduce
computation during predictions for exact top $K$. First, FEXIPRO computes the
thin SVD on the item matrix and uses the output to apply lossless
transformation to the user and item matrices. This transformation places larger
weights in the first few dimensions of each vector, thus improving on the
incremental pruning techniques first suggested by LEMP.  Second, FEXIPRO
applies integer-based quantization to enable more efficient integer-based CPU
instructions for inner product computations. Third, FEXIPRO applies a final
transformation that ensures the inner product will always increase
monotonically by removing negative weights.  Combined, these optimizations yielded
substantial speedups over LEMP \emph{in the point query setting}, in which one
user vector is queried at a time.
}

We include LEMP as well as FEXIPRO~\cite{fexipro} in our study as examples of
state of the art in exact MIPS. Like LEMP, FEXIPRO does not consider blocked
matrix multiply when scoring multiple users. To maximize the generality of our
results, we instantiate our proposed optimization framework for both indexes
(and \simdexi, which we propose), provide a head-to-head comparison of these
indexes using the authors' implementations, and determine the
benefits of adaptive optimization for each in Section~\ref{sec:evaluation}.


\section{A Simple, Hardware-Friendly Index}
\label{sec:index}

Our observation that BMM is competitive with state-of-the-art MIPS indexes on
some---but not all---inputs, provides us with a number of lessons:
\begin{itemize}
\item The optimal choice of MIPS serving strategy is highly variable and
  input-dependent.
\item Hardware efficiency is critical in achieving maximum MIPS performance.
\item We can leverage commodity hardware-optimized
  kernels to improve serving efficiency with limited effort.
\end{itemize}

Based on these insights, we develop a new index that relies heavily on
commodity kernels for hardware-efficient operation, while also pruning the
search space of candidate items. Our goal in developing this index is two-fold.
First, we seek to synthesize a new algorithm that could constructively utilize
the above lessons, combining them with existing ideas from the indexing
literature. Second, we seek to develop a simple-to-understand and
simple-to-implement but efficient index that could serve as a baseline both for
our own empirical comparisons and for others to compare against in the future.
The resulting index, which we call \simdex, relies heavily on commodity
analytics kernels and is implementable in just a few lines of pseudocode, yet
is competitive with (and often exceeds the performance of) author-provided
implementations of LEMP and FEXIPRO.

The remainder of this section introduces \sbi in detail. We first begin with an
overview of the \sbi construction and querying techniques. Subsequently, we
discuss the correctness of this index as well as practical implementation
details and performance considerations.

\begin{figure}[t!]
  \centering
  \includegraphics[width=\columnwidth]{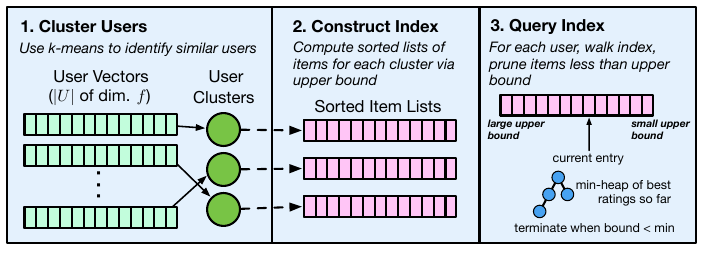}
  \squeezeup
  \caption{Overview of clustering (Section~\ref{sec:clustering}), index
  construction (Section~\ref{sec:construction}), and querying
  (\ref{sec:exact_top_k}) in \simdex.}
  \label{fig:simdex_overview}
  \squeezeup
\end{figure}

\begin{enumerate}
\item \textbf{Cluster Users}: \simdex clusters users into representative
  centroids that will approximate the users' preferences.

\item \textbf{Construct and Query Index}: \simdex computes a
  conservative estimate of the maximum distortion between each
  cluster's predicted rating and the predicted ratings for users in
  the cluster. \simdex subsequently uses this conservative upper bound
  to create a sorted list of items for each cluster.

\item\textbf{Walk Index}: To compute each user's \topK items,
  \simdex walks the item list of the user's corresponding cluster,
  terminating when the previous bound implies there are no
    higher-ranked items to explore.
\end{enumerate}

Figure~\ref{fig:simdex_overview} illustrates these steps, which we proceed to
describe in detail. We first describe \simdex's clustering strategy
(Section~\ref{sec:clustering}). We then show how \simdex uses the cluster
centroids to construct a prediction index (Section~\ref{sec:construction}) and
can subsequently prune item vectors during exact \topk computation
(Section~\ref{sec:exact_top_k}). We conclude with a discussion of optimizations
for performance and runtime analysis (Section~\ref{sec:hw}). Algorithm
\ref{alg:mz_special} provides the entire pseudocode for the core \simdex
routines.

\subsection{Clustering Users}
\label{sec:clustering}
\simdex first partitions the set of users $U \in \Reals^{f}$ into a set of $C$
clusters, with $C \ll |U|$. We will use a user $\u$'s assigned cluster centroid
$\c$ as a means of finding an initial approximation of $\u$'s \topK items,
which we then can iteratively refine per user to find the exact top $K$.

\minihead{Choosing Clusters} As we previously illustrated in
Equation~\ref{eq:upper_bound}, if we wish to provide an upper bound on the true
rating $r_{ui}$ using a centroid $\c$ instead of $\u$, then the quality of our
bound will be determined by the angle between $\c$ and $\u$: the smaller
$\theta_{uc}$ is, the tighter the approximation will be. The tighter the
approximation, the more items we can potentially prune when we query the index
for a user's top $K$. Therefore, we need to choose a clustering algorithm that
ultimately minimizes $\theta_{uc}$.

As prior work has shown~\cite{koenigstein2012efficient}, the ideal clustering
algorithm to minimize $\theta_{uc}$ is spherical clustering, which projects the
centroids onto the unit sphere per iteration and minimizes cosine dissimilarity
rather than Euclidean distance.  However, we found in our experiments that
standard $k$-means reasonably approximates the target goal of minimizing
angular distance, while also being computationally faster than spherical
clustering.  Because $k$-means is a common primitive, there are a number of
widely available, hardware-efficient implementations. In our experiments, we
found that, on average, the set of $\theta_{uc}$s generated by $k$-means were
only 7\% greater than the same set generated by spherical clustering, while
$k$-means was 2-3$\times$ faster, leading to end-to-end runtime improvements of 5-10\%.
Therefore, we use $k$-means for user clustering in \simdexi, and report only
results using $k$-means in Section~\ref{sec:evaluation}.

\subsection{Constructing an Index}
\label{sec:construction}


 \begin{algorithm}[t]
   \scriptsize
   \setstretch{1}
   \DontPrintSemicolon
   \SetAlgoLined
   \SetAlgoNoEnd
   \SetKwFunction{FDot}{Dot}
   \SetKwFunction{FUpdate}{Update}
 \SetArgSty{textup}

 \newcommand\mycommfont[1]{\rmfamily{#1}}
 \SetCommentSty{mycommfont}
 \SetKwComment{Comment}{$\triangleright$ }{}

   \caption{\simdex Index Construction and Querying}

   \SetKwFunction{FBound}{CBound}
   \label{alg:mz_special}
     \Input{\# clusters $n$, users $U$, items $I$}\vspace{.5em}
     \Init{}{
 	  $L \gets$ empty associative array of per-cluster sorted items
         }\vspace{.5em}
     \Fn{\FBound{\textrm{vector $\c$, vector $\mathbf{i}$, angle $\theta_b$}}}{
       \KwRet{$\norm{\mathbf{i}}\cos(\theta_{ic} - \theta_{b})$ \textbf{if} $\theta_{b} < \theta_{ic}$ \textbf{else} $\norm{\mathbf{i}}$} \Comment*[f]{Eqn.~\ref{eq:appx_bound}}}\vspace{.5em}
     \SetKwFunction{FBuild}{ConstructIndex}
     \Fn{\FBuild{}}{
       cluster $U$ into clusters $\mathcal{C}=\{C_1, \dots, C_n\}$ using $k$-means       \;
       \For{cluster $C_j \in \mathcal{C}$ having cluster centroid $\c_j$}{
         $\theta_{bj} \gets \max_{\u \in C_j} \cos^{-1}\left(\frac{\u^T \c_j}{\norm{\u}\norm{\c_j}}\right)$ \Comment*[f]{{max $\theta_{uc}$ for $C_j$}}

         $L[C_j] \gets $ sort $\mathbf{i} \in I$ by \FBound($\c_j$, $\mathbf{i}$, $\theta_{bj}$) descending
         }
       }\vspace{.5em}
       \SetKwFunction{FQuery}{QueryIndex}
       \Fn{\FQuery{\textrm{user $\u$, \topK threshold $K$}}}{
         $H$: min-heap of maximum size $K$ $\gets L[C_j][1:K]$\\\hspace{3mm} s.t. $\u \in $ cluster $C_j$ and each $\mathbf{i} \in H$ is weighted by $\u^T \mathbf{i}$\;
         \For{$\mathbf{i} \in L[C_j][K+1{:}]$}{
           \uIf{\FBound($\c_j$, $\mathbf{i}$, $\theta_{bj}$) $<$ $\min(H)$}{
             \textbf{break}
             }
           \uElseIf{$\u^T\mathbf{i} > \min(H)$}{
             add $\mathbf{i}$ to $H$ with weight $\u^T \mathbf{i}$
           }
         }
         \Return all items in $H$
       }
 \end{algorithm}

As described in Section~\ref{sec:background-indexes}, Koenigstein et
al.  introduced the idea of clustering user vectors to evaluate approximate
\topK queries, and showed how to bound the approximation error based on the angle
between the centroid and user vector (Equation~\ref{eq:upper_bound}).  In this
section, we extend this inequality to compute \emph{exact} \topK queries, which
forms the basis of \simdexi.  In particular, we show how to use this
approximation to prune items that cannot belong in $\u$'s \topK.

Instead of computing the upper bound for each user-item pair, we
instead compute it for the \emph{largest} such $\theta_{uc}$ contained
in a given cluster, which we denote $\theta_b$. Subsequently, we have:

\begin{equation}
  r_{ci}^*\leq
  \begin{cases}
    \norm{\mathbf{i}}\cos(\theta_{ic} - \theta_{b}) & \text{if } \theta_{b} < \theta_{ic} \\
      \norm{\mathbf{i}} & \text{otherwise}. \\
  \end{cases}
\label{eq:appx_bound}
\end{equation}

%
With this coarser approximation,
\simdexi first computes $r_{ci}^*$ for each centroid $\c$ and item $\mathbf{i}$
and, for each centroid $\c$, produces a list of items sorted by
$r_{ci}^*$, denoted $L_c$. 
Both user clustering and this construction procedure are
represented by the \texttt{ConstructIndex} procedure in
Algorithm~\ref{alg:mz_special}.

\subsection{Performing Queries Using Centroid Index}
\label{sec:exact_top_k}

Given the sorted centroid lists, we can now perform \topK queries for
each user. First, we populate the min-heap $H$ with the first $K$ items from
the user's centroid list $L_c$; the items in $H$ are weighted by their
respective true ratings, $r_{ui}$. Then, we examine the remaining items in
$L_c$ by iteratively applying the upper bound in Equation~\ref{eq:appx_bound}.
Each item's upper bound is then compared to the smallest $r_{ui}$ in $H$. If
the upper bound is less than the smallest $r_{ui}$, then that item and all
subsequent items in $L_c$ cannot have an $r_{ui}$ greater than those already in
$H$, since $L_c$ is sorted in descending order by the upper bound, which is
always greater than or equal to $r_{ui}$. Therefore, we can skip those
remaining items and return the items in $H$ as our top $K$. This procedure is
represented by the \texttt{QueryIndex} procedure in
Algorithm~\ref{alg:mz_special}.



\subsection{Hardware-Efficient Execution}
\label{sec:hw}

As described at the start of this section, we designed \sbi to
leverage hardware-efficient libraries for linear algebra and advanced
analytics. \sbi uses $k$-means for clustering, of which there are many
efficient implementations. After computing the upper bound in
Equation~\ref{eq:appx_bound}, \sbi uses efficient sorting routines,
then walks the index during queries. This raises a natural question:
is it possible to hardware-accelerate \sbi's last step of index
traversal?

Because each \sbi cluster is shared across multiple users, we can block the first
several steps of each walk. Specifically, for the first $B$ items in a cluster
list, we perform a blocked matrix multiply between all user vectors in the
cluster and the first $B$ item vectors in the cluster item list. This work
sharing allows \simdex's index traversal routine to make use of more
matrix-matrix multiply operations (instead of less efficient matrix-vector
or inner product operations) while still
benefiting from \simdex's early termination routines. If a user only needs to
visit fewer than $B$ items, this will result in wasted work. However, on
balance, for modest blocking sizes, we find that sharing the first $B$ items is
beneficial to end-to-end runtime.  We evaluate the impact of this optimization
via a lesion study in Section~\ref{sec:evaluation}.

\minihead{Combining Pruning and Hardware Efficiency} In
Section~\ref{sec:evaluation}, we benchmark \simdex on large \U, \I, and $f$,
and yet we still find blocked matrix multiply to be competitive. Nevertheless, a
natural concern with BMM is that, as model sizes grow,
constant-factor runtime improvements due to hardware effects will diminish
compared to an index, rendering matrix multiply much slower. That is, while an
index can potentially prune irrelevant items if they are added to a dataset,
BMM will compute them all. By combining the ability to
algorithmically prune items \emph{and} reap the benefit of hardware-efficient
computation, \simdex mitigates this concern; its index construction routine
will place potentially highly-ranked items at the start of each cluster list,
so the work of many users who are likely to prefer them can be accelerated via
matrix-matrix multiply, without needing to visit \emph{all} items.

\minihead{Index Memory Requirement and Serving Runtime} For a \U and \I with
$f$ latent factors, the \simdex index requires $O(|C||I|f)$ storage, with one
sorted list of length $|I|$ per cluster. Given a $k$-means running time of
$O(f|C||U|i)$ for $i$ iterations and $\bar{w}$, the average number of
items visited per user in Algorithm \ref{alg:mz_special}, \simdex runs
in time:

\begin{equation}
  O(f(|C||U|i +|C||I|\log|I|+ |U|\bar{w}\log K)),
\label{eq:runtime}
\end{equation}

where $|C||I|\log|I|$ captures the index construction time (including
sorting) and $ |U|\bar{w}\log K$ captures the time to walk each
list. \simdex is faster than brute force when
Equation~\ref{eq:runtime} is less than $O(f(|U||I|+|U||I|\log
K))$; therefore, minimizing $\bar{w}$ is instrumental to \simdex
performance.

\minihead{\sbi Parameters} \simdex's index exposes three parameters: the item
blocking factor ($B$), the number of clusters ($|C|$), and the number of
iterations to run $k$-means ($i$). All three of these parameters can be tuned
to maximize performance; however, we found that \simdex's runtime is robust
across various settings of $B$, $C$, and $i$. After conducting a parameter
sweep, we found that $B=4096$, $|C|=8$, and $i=3$ is effective for many inputs.
(Surprisingly, only a few iterations of $k$-means are needed to produce an
adequate set of clusters.) In our evaluation, we report results with these three
settings for all of our experiments.
\subsection{Practical Usage}
Because \simdex clusters users rather than items, it assumes a relatively
static set of users in the desired application. For applications with a dynamic
set of users (assuming an initial static set is present), we can
adapt \simdex as follows: after generating our user clusters on the initial set
of users, we forgo the clustering step for new users and simply assign them to
the extant centroid that yields the smallest $L_2$ distance (\ie perform only
the assignment step in $k$-means). In our experiments, we found this strategy
to be empirically effective for the collaborative filtering datasets used in
our evaluation benchmarks: running $k$-means on a smaller sample of users (10\%) and
then assigning the remaining users to the resulting centroids did not impact the
end-to-end runtime by more than 1\%.  In real applications, the churn in new
users may reach a critical mass, and users can be removed as well as added.
Therefore, periodically scheduling new rounds of user clustering to update the
centroids is an interesting research question, which we leave as future work.

\section{\opt: A MIPS Optimizer}
\label{sec:optimizer}

\begin{figure}[t]
\centering
  \includegraphics[width=0.90\columnwidth]{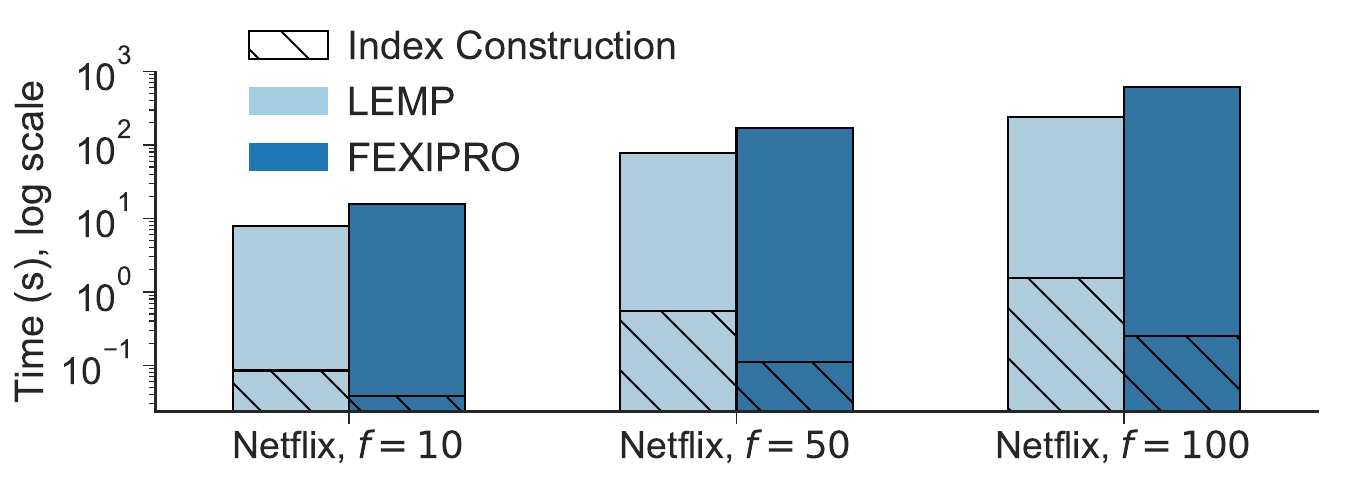}
  \vspace{-0.5em}
  \caption{Index construction time vs. end-to-end runtime for LEMP and FEXIPRO
  retrieving the top $K=1$ items for all users. Index construction time is
  multiple orders of magnitude less than the retrieval time. (Note the log
  scale on the y-axis.)}
  \squeezeup
  \label{fig:index_construction}
\end{figure}


Although \simdex is, on average, an improvement on the state of the art for
MIPS indexes, it still does not always outperform BMM for all inputs. We cannot
rely solely on existing MIPS indexes for the fastest \topK serving strategy.
Further, we cannot use rule-based principles strictly based on the weight
vector characteristics of \U and \I to decide which serving technique is best;
computing a measure of \U and \I's ``index-ability'' is unlikely to be faster
than computing the top $K$ itself. Instead, multiple methods---including
BMM---are necessary to serve a wide variety of inputs as efficiently as
possible, and deciding between these techniques based on empirical cost
measurements is critical.  In response, we propose \opt, a new optimizer that
automates the process of selecting a fast serving strategy.

Given weight vectors $\U$ and $\I$ as input, \opt's goal is to select the
fastest serving strategy, choosing between either an exact indexing strategy or
BMM.  \opt is designed to operate in the online setting,
where we have no a priori knowledge of the input model or underlying
hardware except cache sizing.
Therefore, \opt must run quickly---otherwise, the cost of making an
optimization decision might dominate the cost of computing the top $K$ using a
single technique.

\subsection{Online, Sample-Based Optimization}
When operating online, the key idea in \opt is to estimate the
overall performance of each
serving technique based on a small sample of the users in \U.
Two factors make this strategy possible, and allow \opt to
achieve high accuracy at low (often less than 5\%) runtime overhead.

First, while \topK queries can be slow, requiring up to 14 hours to evaluate
for all of the users in \U in our measurements, the actual index
construction time is relatively short for current indexing
schemes, as is illustrated in
Figure~\ref{fig:index_construction}.  In Section~\ref{sec:evaluation}, we show
that the preprocessing and index construction overhead for FEXIPRO is on
average 1.9\%, for LEMP is on average 0.5\%, and for \simdex
(Section~\ref{sec:index}) is 1.5\% in a batch scoring task that builds an index
and then computes the top $K=1$ prediction for every user. This ratio is even
lower for larger $K$.  Thus, when
evaluating whether to use one of these indexing methods, we can quickly
construct a full index on $\U$ and/or $\I$, and then evaluate its quality using a
sample of the users (as we detail later).  Note that, because index construction
is inexpensive for the fastest indexing techniques in the literature, we
currently always construct a full index.

Second, to evaluate the \topK performance of \simdexb,
\opt extrapolates its performance on a small subset of users.
Optimized matrix multiplication libraries like Intel MKL generally implement an
$O(mnk)$ algorithm (for multiplying $m\times n$ and $n\times k$ matrices)
that scales predictably once matrices are larger than the machine's L2 cache
size, since they divide and process the matrices in blocks.
Thus, \opt runs matrix-matrix multiplication for a subset of users and items that
at least fully occupies the L2 cache and extrapolates performance
from there. As we show in Section~\ref{sec:evaluation}, this approach
accurately captures the performance of BMM.
In addition, we also use the sample to extrapolate the runtime of
extracting the actual \topK values (\eg using a min-heap) once the ratings matrix has been
computed.

Sample-based runtime estimation is a classic technique in data management
systems, with previous applications spanning cardinality
estimation~\cite{cardinality}, online query progress
estimation~\cite{progress}, MapReduce~\cite{parallax}, and cloud~\cite{cloud}
databases. Perhaps closest to our proposal is~\cite{teflioudi2015lemp}, which
uses an online optimizer to select between retrieval algorithms for each bucket of
items. In \opt, we use a sampling-based optimizer to select between using an
index \emph{at all} and BMM---both logical operators for the
same task of \topK search. Additionally, unlike~\cite{teflioudi2015lemp}, we
empirically evaluate the runtime overhead, accuracy in obtaining these
estimates, and effect on end-to-end MIPS serving performance in
Section~\ref{sec:evaluation}.

In more detail, \opt performs online runtime estimation of indexing and BMM as
follows: First, the optimizer constructs an index on \U and/or \I (or multiple
indexes, if necessary). Second, \opt performs queries
with a randomly chosen subset of user vectors and records the runtime. Third,
the optimizer computes the rating predictions via blocked matrix multiply for a
subset of user vectors and computes their top $K$. \opt subsequently estimates
the total runtime for each method.  Finally, if \opt was invoked in a batch
prediction setting, it completes the \topK computation for the remaining users
using the faster approach and reuses the results that it already had for the
sampled users.


To implement the above steps, \opt must also decide on a sample size. For our
target workloads, the number of users is large---at least 480,000. As a result,
we can obtain high-quality estimates with only a small fraction of
users---typically, 0.5\%. However, especially for BMM and for
hardware-optimized indexes such as LEMP, the sample size must be sufficiently
large to demonstrate the benefits of hardware optimizations. For example, if we
perform matrix-matrix multiply with only one user, we effectively perform
matrix-vector multiply, which is substantially slower. Thus, the optimizer must
ensure that the chosen sample is sufficiently large to illustrate hardware
effects; in our optimizer implementation, we require that the sample size at
least occupy the entire L2 cache (256 KB in our experiments). For most \U, this
is easily occupied by a 0.5\% sample---with 64-bit double-precision weights,
2048 \u vectors with 100 latent factors requires 1.64MB of memory, easily
occupying this requirement.  Section~\ref{sec:evaluation} provides empirical
measurements of the effect of sampling on runtime and accuracy.

\minihead{Offline Performance Profiling for BMM}
For blocked matrix multiply, we can alternatively develop an analytical cost model to
predict offline (\ie without sampling) the expected runtime performance of just
\emph{computing the user-item ratings}---but not the selection of the top $K$ items.
Because dense matrix multiply is compute-bound, we can model the runtime based
on the total number of floating-point operations (FLOPs) needed and the FLOPs
per cycle of the CPU~\cite{iakymchuk2011execution}.


%

In our experiments, we found that this analytical model was accurate within 5\% of the
measured dense matrix multiply runtimes in our evaluation. However, this model
does not extend to the \topK selection stage of the MIPS problem: because the
traversal through the min-heap depends on the distribution of ratings for each
user, this data-dependent component of the computation is difficult to
analytically model. While the dense matrix-matrix multiply operation does take the
bulk of the computation runtime, the min-heap traversal time is
non-negligible---at least 9.5\% for our largest models. Therefore, we report
results for \opt only using the online sampling approach in our evaluation.


\minihead{Optimization: Early Stopping with t-test} Rather than measure the
performance of both the index and blocked matrix multiple on the \emph{entire}
sample of users, we can, in some cases, stop early, by applying a one-sample
t-test~\cite{freedman2009statistical} on the per-user query times. That is,
after measuring the performance of blocked matrix multiply, we can
incrementally apply a one-sample t-test on the per-user times seen so far and
compare it against the mean query time provided by BMM. If
the calculated p-value of the test is less than a pre-determined threshold
(e.g., 5\%), then the optimizer can reject the null hypothesis and select
either BMM or the index, whichever has the lower mean query
time.

This technique, of course, will not work for indexes that also batch user
queries for better performance; in those cases, the full sample has to be used
to realize the full effects of the L2 cache. Therefore, \opt cannot employ this
technique for indexes such as \simdexi. However, for indexes that do not batch
users, the t-test is empirically effective for early stopping. For example, we
found in our experiments that \opt needed to only examine 4\% of the full
sample of users when deciding between FEXIPRO and blocked matrix multiply for
$K=1$ on Netflix, $f=10$.

\subsection{Overhead of Optimizer} We can statically bound the overhead of
\opt's optimization routine. Denote the index construction time as $C_I$, the
per-user index-based query time as $Q_I$, and the per-user BMM
query time as $M_I$. Provided that we can estimate $Q_I$ and $M_I$ accurately
for a sample fraction of users of size $s$ (of $n$ total users), then the total
runtime overhead of the optimizer is given by $C_I + \max(Q_I, M_I)
\frac{s}{n}$.

Given that the cost of index construction is on average 1.5\%, we can consider
a few examples. If we use a 1\% sample of users to evaluate our trade-off, and if
$Q_I = M_I$, then the overhead due to optimization is approximately $1\%$. If
$Q_I$ and $M_I$ differ by a factor of $3$, the overhead of optimization
compared to an oracle that automatically selects the right model is
approximately $3\%$. However, compared to simply choosing the slower of the two
methods, the total speedup is $\frac{300}{103} = 2.93\times$. More generally,
given a $d$-times performance differential, the overhead will be $C_I
+\frac{ds}{n}$. Given that the empirical performance differences we observe are
regularly 2-3$\times$, this overhead is relatively small compared to the
benefit.

The above analysis is predicated on the ability to inexpensively and
accurately estimate $Q_I$ and $M_I$. In Section~\ref{sec:evaluation}, we
demonstrate that sampling less than 1\% of users results in high accuracy with
an average of 6.3\% overhead across four types of indexes.

\section{Experimental Evaluation}
\label{sec:evaluation}

In this section, we empirically evaluate the index performance of \simdex, the
runtime efficiency of \simdexb, and the optimizer efficacy of \opt across
various datasets and indexing methods found in the literature.  Our results
show that, across our reference datasets, \opt delivers, on average, a
2.8$\times$ runtime improvement for LEMP, 1.8$\times$ for \sbi, a 5.2$\times$
for FEXIPRO-SI, and 6$\times$ improvement for FEXIPRO-SIR, within 8.8\%,
11.2\%, 11.1\%, and 8.1\% of an oracle optimizer, respectively. In addition,
\sbi is, on average, 1.78$\times$ faster than LEMP and 4.1$\times$ faster than
FEXIPRO; when combined with \opt, \simdex is 3.2$\times$ faster than LEMP.



\subsection{Experimental Setup}

\minihead{Datasets} We use four reference benchmark datasets for our
experimental evaluation (\autoref{table:datasets}), including three
collaborative filtering datasets. Of these, the Netflix dataset is by far the
smallest (480K users and 17K items)---but arguably the best-studied in the
literature for collaborative filtering. The Yahoo R2 dataset is the largest
(1.8M users), and has not been benchmarked in prior work on MIPS serving; we
select it to benchmark \opt's scalability. The fourth dataset, GloVe-Twitter,
contains high-dimensional (up to $f=200$) word embeddings generated from a
corpus of Tweets, and has been previously used to benchmark approximate
nearest-neighbor and MIPS algorithms. Per~\cite{teflioudi2016lempTODS}, we use
the same permutation to select user vectors from the dataset, and use the
remaining vectors as item vectors.

\minihead{Models} We evaluate \topK performance over a range of
collaborative filtering models trained on these datasets, varying the number of
latent factors and tuning the regularization for optimal accuracy.

To begin, the authors of LEMP~\cite{teflioudi2015lemp} have made the
models used in their evaluation publicly
available;\footnote{\url{http://dws.informatik.uni-mannheim.de/en/resources/software/lemp}}
for the Netflix dataset, these models were trained using Distributed
Stochastic Gradient Descent, as described
in~\cite{teflioudi2012distributed}, and we denote these models by
\texttt{*-DSGD} throughout this section. For the Yahoo Music KDD
dataset, the LEMP authors evaluated against the model found
in~\cite{koenigstein2011yahoo}, which is also considered one of the
canonical reference models for this dataset. We denote this model by
\texttt{KDD-REF}.

In addition to these models, we also train explicit feedback models (which
incorporate the ratings made available in each dataset) using the NOMAD
toolkit~\cite{yun2014nomad} (denoted~\texttt{*-NOMAD}). We use the
regularization parameter and hyperparameter settings reported
in~\cite{yun2014nomad} as the starting point for a grid search for the optimal
test RMSE. For the Netflix dataset, we also train an additional set of implicit
feedback models~\cite{implicitfeedback}, using Bayesian Personalized Ranking~\cite{rendle2009bpr} as our training algorithm
(denoted~\texttt{*-BPR}). For the Yahoo R2 Music
dataset,
the literature did not contain any previously reported hyperparameter settings;
therefore, we performed an expanded grid search of $\lambda = \{0, 1e^{-7},
1e^{-6},\dots,1\}$.  For all models (i.e., dataset and number of latent
factors), we utilize the most accurate models with the lowest test RMSE.

\begin{table}
    \centering
\scriptsize

    \caption{Datasets for evaluation}
    \begin{tabular}{| c | c | c | c |}
    \hline
    Dataset & \# users & \# items & \# ratings \\ \hline
    Netflix Prize~\cite{netflix} (Netflix) & 480,189 & 17,770 & 100,480,507 \\ \hline
    Yahoo Music KDD~\cite{dror2012yahoo} (KDD) & 1,000,990 & 624,961 & 252,810,175 \\ \hline
    Yahoo Music R2~\cite{r2} (R2) & 1,823,179 & 136,736 & 699,640,226 \\ \hline
    GloVe-Twitter~\cite{glove} & 100,000 & 1,093,514 & -- \\ \hline
    \end{tabular}
    \label{table:datasets}
    \squeezeup
\end{table}

\minihead{Indexing Strategies} We compare three indexing strategies: LEMP,
FEXIPRO, and \sbi. Recent work from 2016~\cite{teflioudi2015lemp} shows that
LEMP outperforms all prior solutions to the exact MIPS problem , providing a
gold standard for indexes prior to its publication. In a more recent study,
FEXIPRO outperforms LEMP, so we include these two indexing strategies as
representative of the state of the art.

As described below, for LEMP and FEXIPRO, we utilize implementations provided
by the authors of these studies. Our primarily goal is not to directly compare
the engineering quality of these indexes, but instead to understand \emph{i)}
their performance compared to hardware-optimized blocked matrix multiply, and
\emph{ii)} their amenability to online optimization via our proposed methods.
Each index is implemented in C++ with double-precision
floating-point arithmetic.

\miniheadit{LEMP} We utilize the publicly available source code for
LEMP provided by the
authors.\footnote{\url{https://github.com/uma-pi1/LEMP}} We compile LEMP with SIMD
optimizations enabled and tune LEMP's parameters as instructed in the
README. For the retrieval algorithm, we benchmark against LEMP-LI (for
length-based and incremental pruning), which consistently achieves the best
runtimes for \topK computation in~\cite{teflioudi2015lemp}.

\miniheadit{FEXIPRO} We utilize the publicly available source code for
FEXIPRO provided by the
authors.\footnote{\url{https://github.com/Hui-Li/MFRetrieval}} As with
LEMP, we compile FEXIPRO with SIMD optimizations enabled. In addition,
we contacted the authors of FEXIPRO with the experimental results from
this study to obtain guidance regarding parameter tuning. We tuned
parameters per their guidance and report results from both
FEXIPRO with all pruning strategies enabled (denoted FEXIPRO-SIR) and
FEXIPRO with only SVD and integer pruning enabled (denoted
FEXIPRO-SI).

\miniheadit{\sbi Implementation} We implement \simdex in
C++\footnote{\url{https://github.com/stanford-futuredata/optimus-maximus}} using double
precision. We use the open-source Armadillo library~\cite{armadillo} for
$k$-means and Intel MKL~\cite{intel-mkl} for \sbi's shared index traversal. We also compared with
OpenBlas~\cite{openblas} and found limited effect on runtime for the relatively
small shared index traversal operation.  

\miniheadit{BMM Implementation} We also compare to brute-force blocked matrix
multiply using Intel MKL~\cite{intel-mkl}. We compute ratings for users in a
series of batches that each occupy the entirety of memory, and use a min-heap
from the C++ standard library to compute the true top $K$.

\minihead{Optimizer Implementation} We implement \opt from
Section~\ref{sec:optimizer} as a subroutine within each of the three indexes'
main routines. \opt performs sampling and runtime estimation, and then invokes
either MKL or the regular indexing routine based on the output. 

\minihead{Environment} We report results from an Intel Xeon E7-4850 v3 2.20 GHz
processor with 1 TB of RAM. We allow Intel MKL to use the entire RAM in
blocking. Unless otherwise noted, we evaluate our benchmarks
on a single core and report results in the single-threaded setting.

\subsection{End-to-End Index Performance}
\label{sec:bakeoff}

\begin{figure*}[ht!]
  \centering
  \includegraphics[width=\textwidth]{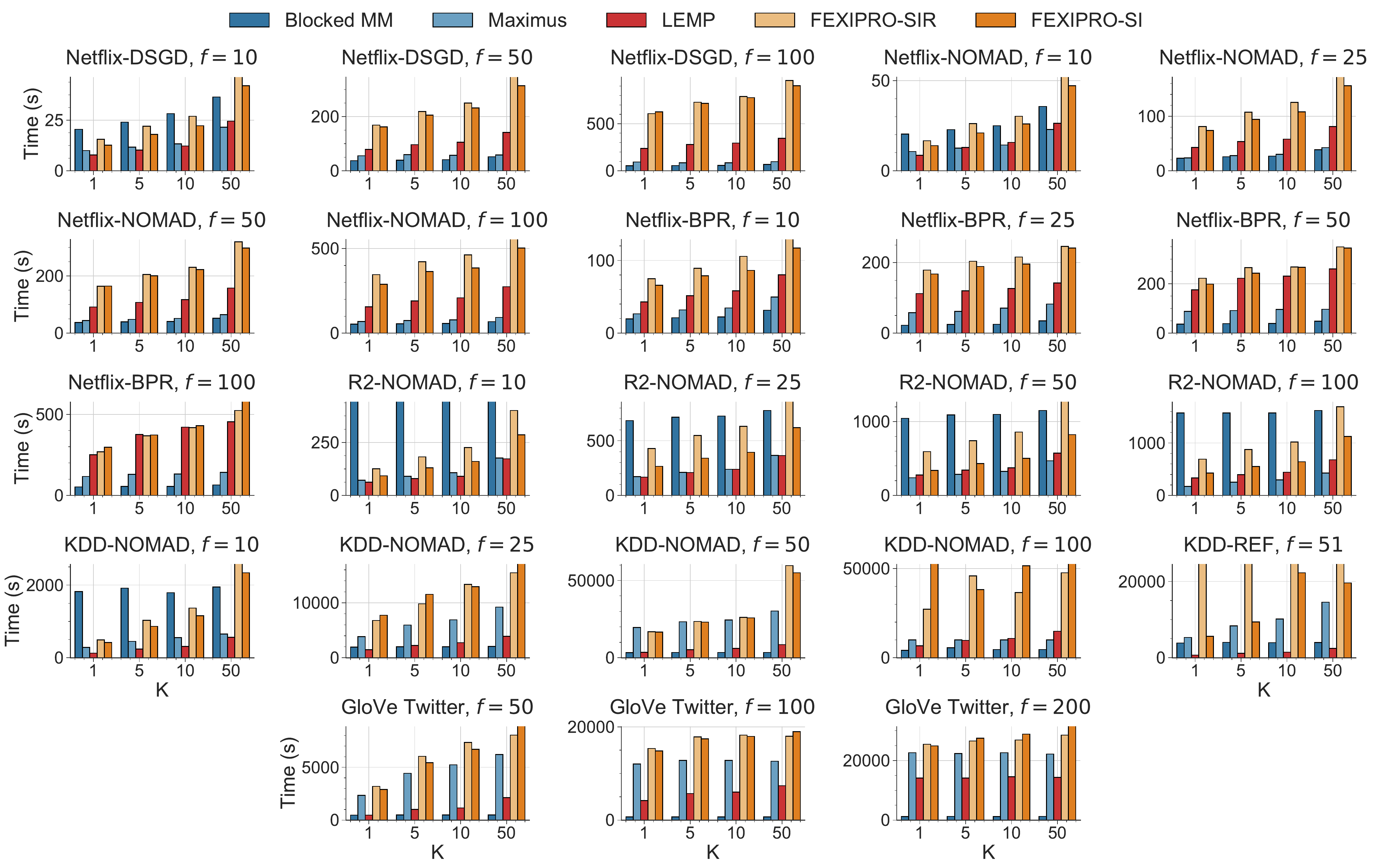}
  \squeezeup
  \caption{MIPS wall-clock time of \simdex and BMM compared to LEMP and FEXIPRO
  (lower is better).
  When combined with \opt, \simdex is 3.2$\times$ faster than
  LEMP, 7.3$\times$ faster than FEXIPRO-SIR, and 6.3$\times$ faster than
  FEXIPRO-SI.}
  \vspace{-1.5em}
  \label{fig:bar_charts}
\end{figure*}
In this section, we examine the end-to-end runtime---which includes index
construction time---of \sbi compared to LEMP, FEXIPRO, and BMM to compute
the top $K$ for all users. Figure~\ref{fig:bar_charts} depicts the results
across all four datasets and all 23 reference models; we report the runtimes
for $K = \{1, 5, 10, 50\}$.  Without using \opt, \sbi is, on average,
1.8$\times$ (and up to 10.6$\times$) faster than LEMP and  $>
10\times$ faster than both FEXIPRO-SI and FEXIPRO-SIR. Head to head, \sbi is
faster than LEMP 67\% of the 92 model/top-$K$ combinations. Compared to
FEXIPRO-SI, \sbi is faster for every single combination in our benchmarks
except for one: KDD-NOMAD, $f=50$, $K=5$.

These results stand in contrast to ~\cite{fexipro}, which reports
significant speedups over LEMP. We have presented and discussed these
results with the authors of both LEMP and FEXIPRO and, after conferring with them,
believe these differences are due to several factors. First, unlike
LEMP and \simdex, FEXIPRO is optimized for the point query setting, in
which a single user's top $K$ is queried at a time. Thus, FEXIPRO does
not take advantage of hardware blocking, which could yield additional
performance benefits. In addition,~\cite{fexipro} reports runtime
relative to a custom implementation of LEMP, as opposed to the LEMP
authors' implementation that we benchmark here. Finally, because the
results presented here utilize models that have been regularized
to obtain the lowest test RMSE, with the exception of the
\texttt{*-DSGD} models, these models are not the same as those
appearing in prior publications. 

Compared to \simdexb, \sbi is 2.7$\times$ faster on average, but this speedup is not
present for all models and values of $K$; \simdexb is actually faster in
34.8\% of the 92 model/top-$K$ combinations in our benchmarks.  (Note
that the runtime for \simdexb varies with $K$, due to the time necessary to
traverse the min-heap for all users.) The difference in runtime can also
vary widely between the two: \sbi can be up to 43.4$\times$ faster than
BMM for our reference models, but BMM can also be up to 18.7$\times$
faster than \sbi. This underscores the need for an optimizer that can choose
between \simdexb and \sbi (or other indexes) to achieve the best performance for MIPS.

Finally, note that, between \simdexi, BMM, and LEMP, no
pair of techniques yielded the fastest runtime on all the reference models.
LEMP is fastest on 11 of the 92 combinations, \simdexb is fastest on 53 of
them, and \simdexi is fastest on the remaining 28. This indicates that a
three-way optimizer that decides between LEMP, \simdexi, \emph{and}
BMM is ideal. We investigate this question in Section~\ref{sec:optimizer_eval}.

\minihead{Multi-core Experiments}
In addition to these single-threaded benchmarks, we also ran multi-core
experiments to determine whether or not \simdexi and prior MIPS indexes would
scale with increased parallelism. While \simdexb is trivial to parallelize,
parallelizing the index structures of each respective system may not be as
straightforward. The results are summarized in
Figure~\ref{fig:parallel_experiment}: both \simdexi and LEMP achieved a
near-linear speedup as we increase the number of cores from 1 to 16. Because
both indexes are read-only, a simple partitioning scheme across users proves to
be an effective parallelization strategy.

\begin{figure}[ht!]
  \centering
  \includegraphics[width=0.95\columnwidth]{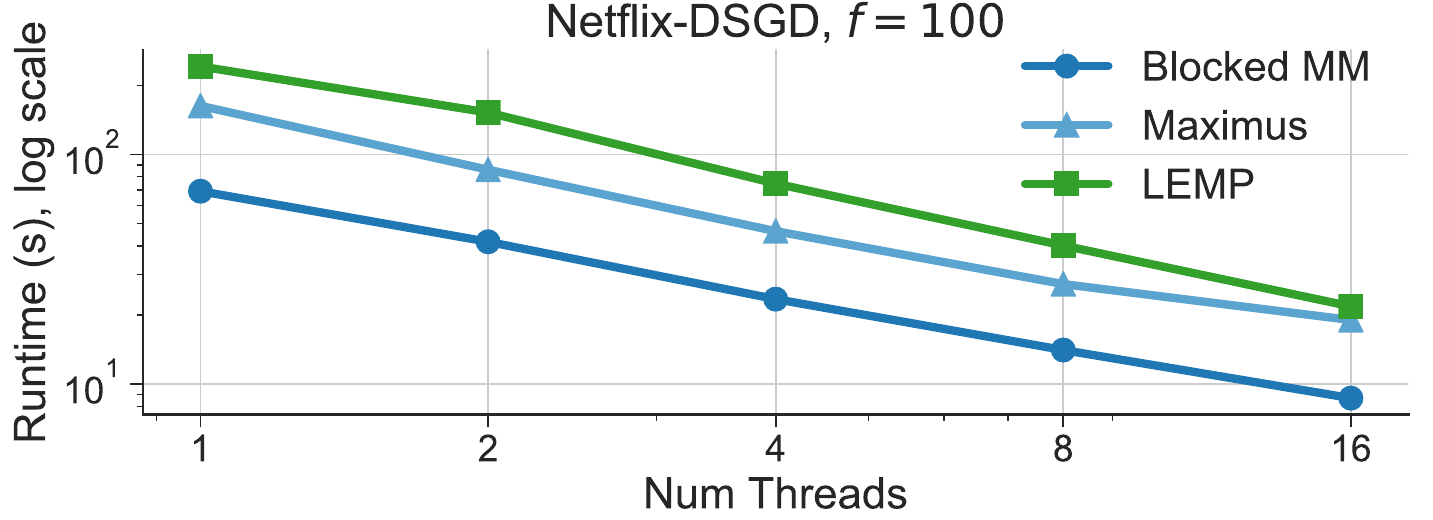}
  \vspace{-1.0em}
  \caption{End-to-end runtime of $K=1$ query benchmarked on \simdexb,
  \simdexi, and LEMP across multiple cores. (FEXIPRO did not provide a
  multi-core implementation.) Both \simdexi and LEMP achieve near-linear
  speedups as we increase the number of cores.}
  \vspace{-0.5em}
  \label{fig:parallel_experiment}
\end{figure}

\subsection{\opt Efficacy and Overhead}
\label{sec:optimizer_eval}

\begin{table*}[]
  \scriptsize
\centering
\caption{Effectiveness of online optimizer on reference models}
  \begin{tabular}{@{}lllllll@{}}
\toprule
                  &          &               &                            & \multicolumn{3}{c}{Avg. speedup versus LEMP-only Baseline}   \\
Optimizer Choices & Accuracy & Avg. Overhead & Std. Dev. Overhead         & Index Only   & \opt (w/ overhead) & Oracle (no overhead) \\ \midrule
BMM + LEMP         & 89.1\%   & 4.3\%         & \multicolumn{1}{l|}{4.2\%} & 1$\times$    & 2.81$\times$           & 3.08$\times$         \\
BMM + FEXIPRO-SI   & 97.8\%   & 6.4\%         & \multicolumn{1}{l|}{8.1\%} & 0.50$\times$ & 2.60$\times$           & 2.93$\times$         \\
BMM + FEXIPRO-SIR  & 97.8\%   & 6.4\%         & \multicolumn{1}{l|}{7.8\%} & 0.43$\times$ & 2.56$\times$           & 2.88$\times$         \\
BMM + \simdexi           & 93.5\%   & 5.5\%         & \multicolumn{1}{l|}{5.9\%} & 1.78$\times$ & 3.15$\times$           & 3.43$\times$         \\
BMM + LEMP + \simdexi    & 84.8\%   & 9.1\%         & \multicolumn{1}{l|}{8.4\%} & -            & 2.99$\times$           & 3.48$\times$         \\ \bottomrule
\end{tabular}
\label{table:optimizer}
\squeezeup
\end{table*}

To demonstrate the effectiveness and generality of \opt, we combine it with
each of the three index techniques---\simdexi, LEMP, and FEXIPRO-SI/SIR---to
choose between the given index and BMM on the 92 model/top-$K$
combinations benchmarked in Figure~\ref{fig:bar_charts}.  We also include an
additional experiment in which we pair \opt with both \simdexi and LEMP, thus
using \opt to perform a three-way optimization.

To effectively compare \opt's runtime for a target model and $K$, we
consider two baselines: first, we normalize our results by the runtime
of the LEMP index only. Second, we consider the runtime resulting from
consulting an oracle optimizer that always chooses the fastest
strategy without incurring any runtime overhead. Table~\ref{table:optimizer}
summarizes the results of these experiments: for each two-way indexing
strategy, \opt improves the \topK runtime and is within 8.3\% of the
optimal runtime efficiency compared to the oracle (average: 5.1\%,
including overhead). To observe the benefits of \opt, consider
FEXIPRO-SIR, the slowest index that we measure in our
benchmarks. Without \opt, FEXIPRO-SIR is, on average, 2.3$\times$
slower than LEMP on our reference models; however, with \opt, it
becomes 2.6$\times$ faster than LEMP, within 11\% of the
maximum improvement that an oracle would provide. The
other indexes exhibit similar trends, including \simdexi, which becomes
3.2$\times$ faster than LEMP---within 9\% of the max---when paired
with \simdexb using \opt.

The Accuracy measurement in Table~\ref{table:optimizer} refers to
\opt's classification accuracy: how often does it choose the fastest
serving strategy?  For FEXIPRO and \simdexi, \opt effectively makes
the correct choice with greater than 93\% accuracy: in the case of
\simdexi, only six model/top-$K$ combinations are misclassified,
while only two are misclassified for FEXIPRO-SI/SIR. These
misclassifications do not lead to a significant runtime penalty: for
the six misclassified examples for \simdexi, the runtimes of
\simdexi and \simdexb are within a few seconds.

However, the accuracy for LEMP is less impressive: 10 of the 92
combinations are misclassified. To understand why this was the case,
we ran an experiment in which we varied the user sample ratio used in
\opt and measured the estimated runtime of the serving strategy
(either \simdexb or an index) based on that sample. We measured the
variance of the overall estimates for all of the indexes and BMM.

\begin{figure}[t!]
  \centering
  \includegraphics[width=0.95\columnwidth]{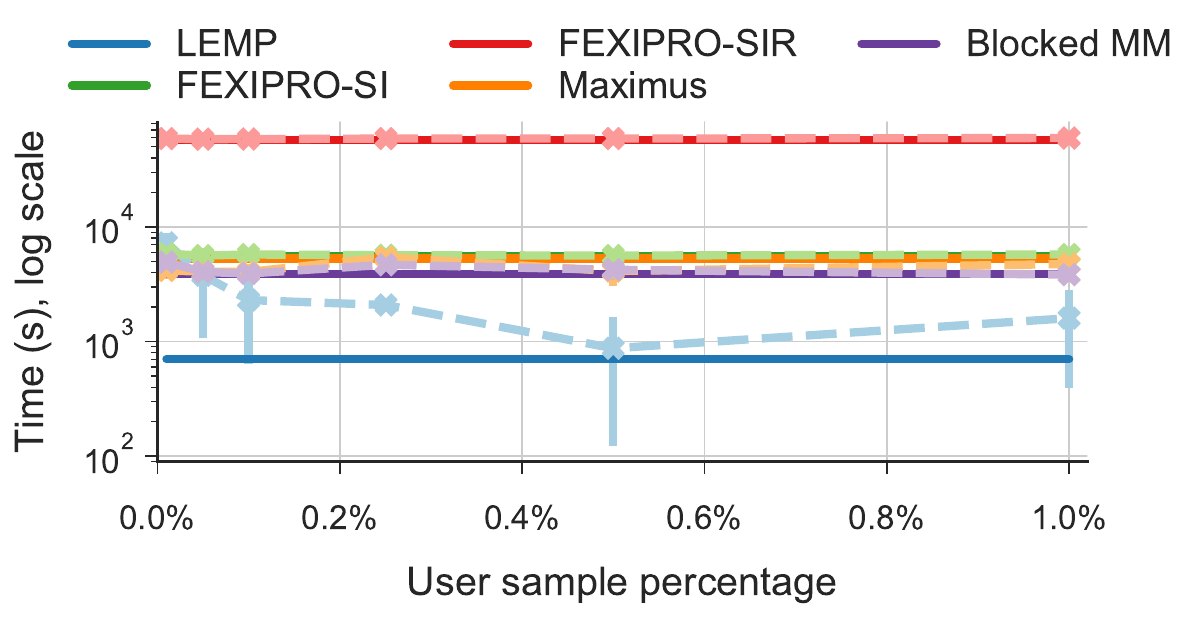}
  \vspace{-1.0em}
  \caption{\opt runtime estimates for all indexes on KDD-REF, $f=51$, $K=1$.
    Solid lines represent the true runtimes; dashed lines denote the estimates.
    Estimates are averaged over four runs; error bars
    represent the standard deviation. Despite the high variance of the
    estimates for LEMP, \opt makes the correct choice with a small ($<
    1\%$) sample of users.}
  \vspace{-1.0em}
    \label{fig:runtime_estimate}
\end{figure}

Figure~\ref{fig:runtime_estimate} summarizes the results of our
experiment on KDD-REF, $f=51$, $K=1$ for five distinct sample ratios
ranging logarithmically from 0.01\% to 1\%. As depicted, \opt's user
sampling technique is robust and exhibits relatively low variance for \simdexi,
BMM, and FEXIPRO, but the estimated runtimes for LEMP have
much higher variance. This is because LEMP performs runtime
adaptation of its serving strategy, and two separate samples of users
may result in different pruning strategies (i.e., coordinate-wise
versus $L_2$-wise pruning; see Section~\ref{sec:background-indexes}). For this
particular model and choice of $K$, \opt is still able to make the
right decision, since the estimated runtime for LEMP never exceeds the
runtime estimate (or the true runtime) for BMM. However, for
other models, this is not always the case, which explains why
\opt's accuracy with LEMP is lower. Nevertheless, combining \opt with
LEMP is still beneficial, yielding a 2.8$\times$ speedup that is
within 9\% of the speedup obtained by the oracle.

The high variance exhibited by LEMP in our experiments also
demonstrates why the three-way optimizer in the bottom row of the
table---BMM + LEMP + \simdexi---actually achieves a smaller overall
speedup (3$\times$) than BMM + \simdexi (3.2$\times$). In addition
to the slight reduction in accuracy (84.8\%), \opt begins to incur a
higher runtime overhead---9.1\% on average---since it now has to
construct and query multiple indexes. For this three-way comparison
\opt's speedup is within 15\% of the max possible speedup,
3.48$\times$.

We were initially surprised to find that absolute optimizer accuracy
was not a robust signal of end-to-end speedup. The three-way optimizer
demonstrates this well: with 84.8\% accuracy, \opt is still within
15\% of the optimal speedup on average. The primary reason for this
phenomenon is that, in the cases where \opt chooses a sub-optimal
strategy, the margin of difference is often small. For example, in
R2-NOMAD $f = 50,\ K=1$ (Figure~\ref{fig:bar_charts}), \simdex and
LEMP are remarkably close---within 12\%---while each index is
considerably faster than BMM---a factor of
3.75$\times$. Thus, unless \opt's runtime estimate is off by more than
3.75$\times$, choosing the slower index has limited impact on this
model. This result suggests that even coarse-grained runtime estimates
are useful in accelerating these indexing workloads.

\minihead{Runtime Analysis of \simdex and Item Blocking Lesion Study} Finally, to
better understand the impact of each stage of \simdex's execution, we measured
the running time of each component. On average, \simdex incurs an overhead of
1.8\% for clustering, constructing its index, and performing cost estimation.
We illustrate a breakdown for Netflix-NOMAD, $f=50$ and R2-NOMAD, $f=50$ in
Figure~\ref{fig:factor_analysis}, both of which use \simdex's index; for
Netflix, \simdex spends 0.79 seconds in the first three stages, and over 43
seconds in the final stage when computing predictions. We also
illustrate the effect of hardware-efficient item blocking, which delivers
2.4$\times$ and 1.4$\times$ speedups for these datasets; the effect for Netflix
is more pronounced because the average number of items visited in the index for
each user (i.e., $\bar{w}$) is larger than in Yahoo R2. Sharing a single, small
blocked matrix multiply at the start of index traversal allows \simdexi to
benefit from hardware-efficient BLAS. Overall, \simdex's overheads are small,
especially relative to the speedups that \opt enables by
choosing between \simdexi and \simdexb.

\begin{figure}[t!]
  \centering
  \includegraphics[width=0.95\columnwidth]{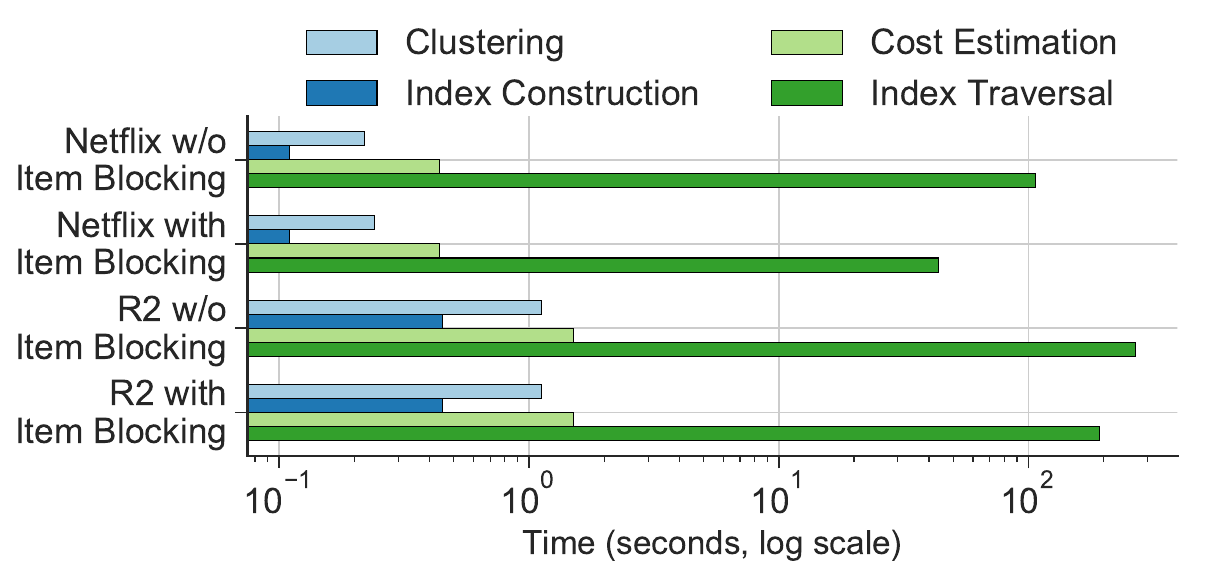}
  \vspace{-1.2em}
  \caption{Runtime breakdown of \simdex for $K=1$ on Netflix-NOMAD,
  $f=50$ and R2-NOMAD, $f=50$. Enabling item blocking improves throughput by
  2.4$\times$ and 1.4$\times$, respectively.}
  \squeezeup
  \label{fig:factor_analysis}
\end{figure}

\section{Related Work}
\label{sec:related_work}



\minihead{State of the Art: LEMP and FEXIPRO}
LEMP~\cite{teflioudi2015lemp} (previously described in
Section~\ref{sec:background-indexes}) and FEXIPRO~\cite{fexipro} are the
closest related work in exact MIPS, which we build upon and evaluate against in
this work. Here, we first summarize the key contributions of FEXIPRO, then
discuss key differences between these two indexes and our work.

FEXIPRO~\cite{fexipro} leverages three pruning strategies for the \topk MIPS
problem: \emph{i)} applying singular value decomposition to the input to only
compute partial inner products, \emph{ii)} applying input quantization to
perform integer-based operations only, and \emph{iii)} applying a non-negative
transformation on the input to ensure monotonicity, further promoting pruning.
In their paper, the FEXIPRO authors report that these three strategies combined
together achieve an order-of-magnitude speedup over a custom point-query
implementation of LEMP. As discussed in Section~\ref{sec:bakeoff}, our results
differ on the models we evaluated.

There are several differences between the prior work on MIPS serving and ours.
First, both LEMP and FEXIPRO do not consider bypassing the index search
altogether to use blocked matrix multiply, which we show is often better than
their performance. Neither LEMP nor FEXIPRO considers BMM as a strategy when
scoring a group of users, and their data structures are not designed to enable
BMM during index traversal.  While LEMP samples users to select the best
retrieval algorithm per bucket, \opt's sampling strategy answers a much coarser
granularity question: should a given model be indexed at all? To minimize the
overhead of sampling, we apply an incremental t-test for early stopping,
another difference between our approach and LEMP's. Lastly, both LEMP and
FEXIPRO build an index on the items, rather than on the users.

\minihead{User Clustering and Other Indexes} Several alternative methods
propose tree-based indexes for MIPS.  As discussed in Section~\ref{sec:background-indexes},
Koenigstein et al.~\cite{koenigstein2012efficient} introduce a method for
computing \emph{approximate} \topK recommendations via user clustering. In
their method, the user vectors are clustered using spherical clustering and the
top $K$ items are pre-computed for each cluster---the \topK for a given user is
the \topK of the cluster the user belongs to. In their analysis, they provide the
bound we use in Equation~\ref{eq:upper_bound} in Section~\ref{sec:index}, but
they use it for the purpose of \emph{approximate}---not exact---\topK. In
\simdex, we show how to utilize this bound to build an efficient index for
exact MIPS.

Ram and Gray~\cite{ram2012maximum} present three different techniques for
indexing item vectors: single-ball trees, dual-ball trees, and cone trees. All
three data structures share a similar strategy: they recursively subdivide the
metric space of items into hyperspheres. Every node in the tree represents a
set of points, and each node is indexed with a center and a ball enclosing all
the points in the node. Of the three, the cone tree offers the fastest
speedups.  Follow-on work by Curtin et al.~\cite{curtin} extends this method
using cover trees~\cite{covertree}, but Teflioudi et
al.~\cite{teflioudi2015lemp} show that these methods are slower than LEMP.


\minihead{Approximate MIPS}
A large body of work considers the approximate setting for MIPS, whereby search
procedures attain approximations of the true top $K$.  Shrivastava et
al.~\cite{shrivastava2014asymmetric} use asymmetric hash functions in their LSH
subroutine, which reduce the approximate MIPS problem to a sublinear
nearest-neighbor search. Similarly, Bachrach et al.~\cite{bachrach2014speeding}
also reduce MIPS to NNS using a novel Euclidean transformation,
which allows users to trade off \topK accuracy for better performance. Our
focus in this work is the exact setting---delivering the most accurate
predictions possible with the fastest speed.

\minihead{Model Serving} Model serving systems are of increasing practical
importance in the database community~\cite{kumarsurvey,baylor2017tfx}.
Prior systems, such as \textsc{TuPAC}~\cite{sparks2015automating},
have shown that dense matrix multiply is an effective approach for
model/parameter search, while other systems, such as
Clipper~\cite{crankshaw2017clipper}, demonstrate that adaptive batching can
lead to performance wins in online serving.  In our work, we show that,
\emph{for certain models}, blocked matrix multiply is actually faster than
state-of-the-art MIPS indexes, which prune significant computation but do not
take advantage of hardware efficiency. Thus, we introduce a new MIPS index,
\simdex, that takes advantage of high-performance BLAS libraries while also
pruning computation.  Because the best serving strategy varies from model to
model, we introduce \opt, which determines the best serving strategy using an
online, sampling-based approach.



\section{Conclusion}
\label{sec:conclusion}

MIPS is a critical component of many modern workloads, including recommender
systems and information extraction tasks. In this work, we show that the
fastest of today's indexes do not always outperform blocked matrix multiply. We
thus propose \simdex, a simple but efficient indexing scheme that leverages
linear algebra kernels to gain hardware efficiency while also pruning
computation.  In addition, we design \opt, a system that can efficiently choose
between using an index and blocked matrix multiply.  Together, \opt and \simdex
achieve speedups of 3.2$\times$ on average, and up to 10.9$\times$, on popular
word embeddings and well-tuned models from recommendation datasets.

\section{Acknowledgments}
This research was supported in part by affiliate members and other supporters
of the Stanford DAWN project---Ant Financial, Facebook, Google, Infosys, Intel,
Microsoft, NEC, SAP, Teradata, and VMware---as well as Toyota Research
Institute, Keysight Technologies, Hitachi, Northrop Grumman, Amazon Web
Services, Juniper Networks, NetApp, and the NSF under CAREER grant CNS-1651570.

\bibliographystyle{abbrv}
\bibliography{refs}

\end{document}